\documentclass[11pt]{article}

\title{Asynchronous iterations in ultrametric spaces}
\author{Alexander J. T. Gurney}
\date{\today}

\usepackage[margin=1.2in]{geometry}
\usepackage{amsmath}
\usepackage{amsthm}
\usepackage{amsfonts}
\usepackage[sort,numbers]{natbib}
\usepackage{fourier}

\newtheorem{theorem}{Theorem}
\newtheorem{lemma}[theorem]{Lemma}

\theoremstyle{definition}
\newtheorem{definition}{Definition}
\newcommand{\newterm}[1]{\emph{#1}}

\newcommand{\setst}[2]{\ensuremath{\left\{#1\;\middle|\;#2\right\}}}
\newcommand{\set}[1]{\ensuremath{\left\{#1\right\}}}
\newcommand{\setsize}[1]{\ensuremath{\left|#1\right|}}
\newcommand{\ball}[2]{\ensuremath{B(#2; #1)}}

\newcommand{\norm}[1]{\ensuremath{\left\|#1\right\|}}

\newcommand{\Time}{\ensuremath{T}}

\newcommand{\NN}{\mathbb{N}}
\newcommand{\RRpos}{\mathbb{R}_{\mathord{\ge} 0}}
\newcommand{\QQpos}{\mathbb{Q}_{\mathord{\ge} 0}}

\begin{document}

\maketitle

\begin{abstract}
Some iterative calculations can be carried out by parallel
communicating processors, and yield the same results whether or not
the processors are synchronized.
We show that this is the case if and
only if the iteration is a contraction that is strict on orbits,
with respect to an ultrametric defined on the state space.
The maximum number of independent processors is given by the dimension
of the space.

We apply this theorem to interdomain routing, and are able to provide
two advances over the previous state of the art.
Firstly, multipath routing problems have unique solutions, if certain
conditions are satisfied that are analogous to known correctness
conditions for the single-path case.
Secondly, these solutions can be computed asynchronously in a variety
of ways, which go beyond methods that are currently used.

\end{abstract}

\clearpage

\section{Introduction}
\label{sec:intro}

The theory of \emph{asynchronous iterations}
is concerned with the problem
of when an iterative algorithm
can be implemented on a set of communicating
processors, without explicit synchronization,
and yet still compute the same result.
It is known that
for this to be possible,
certain characteristics of the iteration
must hold with respect to its state space:
several different sufficient conditions are known.
These are special cases of a more general result,
which gives a necessary and sufficient condition
for asynchronous safety.
It requires that the state space have a `nested box' structure,
and that synchronous iterations always lead to a more-inward box.
This condition is rather `low-level', and may be difficult to
verify in many cases; equally, the various sufficient conditions
are easier to work with, but do not account for all possibilities.
In addition, much of the prior work on these iterations
assumes that data values are real numbers,
whereas there are many iterative algorithms
that work over other kinds of data.

In this paper,
we reinterpret the `nested box' structure
in terms of a special kind of metric on the state space.
In such an \emph{ultrametric}, the balls around a given point always form
nested boxes.
The Banach fixed point theorem for ultrametric spaces,
which states that a self-map of the space that is contractive and
strict on orbits must have a unique fixed point,
is precisely the theorem needed to prove
asynchronous safety.
That is, the application of this theorem proves not only that there is
a unique fixed point, but also that it can be found by asynchronous
iteration.
Conversely, whenever an iteration is asynchronously safe,
then an ultrametric can be defined with respect to which the iteration
is a contraction of the required kind.
Furthermore,
the degree of potential asynchrony (the number of processors across
which the iteration can be partitioned) is given by the dimension of
the ultrametric space.
This result applies to discrete data as well as to numeric problems.

In the final part of the paper, we apply this new theorem to a problem
in interdomain multipath routing. The presence of a unique fixed
point, let alone the possibility of asynchronous implementation, was
not previously known. Existing sufficient conditions did not cover
this case, but it is dealt with by the new theorem.

\section{Background}
\label{sec:background}

This section explains the 
two separate areas of theory---asynchronous iterations
(Section~\ref{ssec:async})
and ultrametric spaces (Section~\ref{ssec:ultra})---which
are involved in the main results of this paper.

\subsection{Asynchronous iterations}
\label{ssec:async}

There are many algorithms which operate by iteratively applying the
same function to some state vector.
If the state space is $M = M_1 \times M_2 \times \cdots \times M_k$,
then a function $\sigma$ from $M$ to $M$ can be decomposed as the
product of $k$ functions $\sigma_i : M \rightarrow M_i$, where
\[
   \sigma(m_1, m_2, \dots, m_k)
     = (\sigma_1(m_1, m_2, \dots, m_k),
        \sigma_2(m_1, m_2, \dots, m_k),
        \dots,
        \sigma_k(m_1, m_2, \dots, m_k) ).
\]
On a sequential machine, each $\sigma_i$ function must be evaluated in
turn in order to produce the new state vector. But if multiple
processors are available, then each function evaluation can be
assigned to a different processor; once they have all finished, they
can mutually communicate their results so that the next iteration step
can begin. The total execution time for a single iteration is thus
bounded by the time taken for the slowest processor.

It has long been known that for \emph{some} algorithms, synchronous
execution is not required for convergence. That is, the same answer as
produced by the above process can also be generated in a far less
restrictive execution model. In the asynchronous execution of the same
algorithm, the same function is executed on each processor, but the
execution is no longer in lock-step and the input data may come from
a prior round of the iteration.

Asynchronous iteration has been applied
to
many problems:
finding shortest paths~\cite{TsengBertsekasTsitsiklis90},
dynamic programming~\cite{BertsekasTsitsiklis89},
finding fixed points of linear and non-linear operators~\cite{Baudet78},
PageRank~\cite{KolliasGallopoulosSzyld05},
and several others~\cite{FrommerSzyld00}.
Use of the method is often motivated by the
large quantity of data involved, or by
the difficulty in ensuring synchronized execution.
Sometimes, particularly for numeric problems, convergence time can
be improved by dropping the synchrony requirement.
While the theory was originally developed
for numeric iterations~\cite{Baudet78,ChazanMiranker69,Bertsekas83},
results have also been obtained
for iterations on discrete data~\cite{UresinDubois90}.

In the following, we take time to be discrete and linear; the set $T$
contains all time values.

\begin{definition}
For a set $P$ of processors, an \newterm{asynchronous execution
schedule} consists of two functions $\alpha$ and $\beta$, where
\begin{itemize}
 \item $\alpha : \Time \rightarrow 2^P$ yields the set of processors
       which activate at each time step, and
 \item $\beta : \Time \times P \times P \rightarrow \Time$ yields the
       delay between two given processors at each time step; so
       if $\beta(t, i, j) = t'$, then the data from $j$ used at $i$ at
       time $t$ was generated at time $t'$.
\end{itemize}
\end{definition}

\begin{definition}
A schedule $(\alpha, \beta)$ on $P$ is \newterm{admissible} if
\begin{enumerate}
\item For all $i$ in $P$, and $t$ in $\Time$, there exists $t' > t$
      such that $i$ is in $\alpha(t')$.
\item For all $t$ in $T$, and $i$ and $j$ in $P$, $\beta(t, i, j) > t$.
\item For all $i$ and $j$ in $P$, and $t'$ in $\Time$, there exists a
      $t_f$ in $\Time$ such that if $t > t_f$, then $\beta(t, i, j)
      \neq t'$.
\end{enumerate}
\end{definition}

These admissibility conditions may be expressed more informally, as:
every node activates infinitely often;
information does not propagate backwards in time; and
a past data value can only be used finitely often.
Some weaker versions of these conditions have also been considered;
for example, the final axiom may be replaced by an upper bound on the
age of any data item used in a calculation~\cite{UresinDubois90}.

If $\sigma$ is an function from $M$ to $M$, then we can define an
asynchronous iteration corresponding to $\sigma$ for any given
schedule, and for a particular starting point in $M$.
For each $i$, there will be a series of values in $M_i$
generated by the iteration; call these $x_i(t)$ for $t \in T$.
Let $x(t)$ be their product, so $x(t)$ is a vector in $M$.

Let $m$ be a point in $M$,
and
let $(\alpha, \beta)$ be a schedule on a set of $k$ processors.
For each $i$, let $x_i(0) = m_i$.
For $t > 0$ we
define $x_i(t)$ by
\[
   x_i(t) = 
      \begin{cases}
        x_i(t-1) & i \not\in \alpha(t) \\
        \sigma_i\left( 
            x_1(\beta(t, i, 1)), x_2(\beta(t, i, 2)), \dots,
            x_k(\beta(t, i, k)) \right) & i \in \alpha(t).
      \end{cases}
\]
So if processor $i$ does not update at time $t$, its value does not
change; when it does update, it carries out its usual $\sigma_i$
operation, but may use values from further in the past than the
immediately preceding step.

Note that if $\alpha(t) = \set{1, 2, \dots, k}$ for all $t$,
and $\beta(t, i, j) = t - 1$ for all $t$, $i$ and $j$,
then this is just the synchronous iteration from above.

\begin{definition}
Let $M = \prod_{i \in I} M_i$, for some sets $M_i$.
Let $\sigma$ be a function from $M$ to $M$ that has a unique fixed
point $m^{\ast}$.
Then $\sigma$ is an
\newterm{asynchronously contracting operator} (ACO)
if, for any admissible schedule $(\alpha, \beta)$ on $\setsize{I}$,
and any starting point $m$ in $M$,
there is some time $T_{m,\alpha,\beta}$
such that for any $t > T_{m,\alpha,\beta}$,
the state $x(t)$ is equal to $m^{\ast}$.
\end{definition}

These ACOs may also be characterized in terms
of a `nested box' structure on the state space.
Informally, if the synchronous iteration is such
that it always takes points into a more-inward box,
then it is asynchronously safe.
The asynchronous iterations will eventually lead to the
same fixed point, but may deviate from the synchronous
course of execution.

\begin{definition}
A subset $N$ of $M = \prod_{i \in I} M_i$ is a \newterm{box}
if, for each $i$, there is a subset $N_i$ of $M_i$ such
that $N = \prod_{i \in I} N_i$.
\end{definition}

\begin{theorem} \label{thm:aco-box}
An operator $\sigma$ on $M = \prod_{i \in I} M_i$ is an ACO 
if and only if
there exist boxes $\set{C_0, C_1, \dots, C_k}$ in $M$
with the following properties:
\begin{enumerate}
\item $C_0 = \set{m^{\ast}}$ for some $m^{\ast}$ in $M$.
\item $C_k = M$.
\item If $0 \le r < s \le k$, then $C_r \subset C_s$.
\item If $m$ is in $C_{r+1}$, then $\sigma(m)$ is in $C_r$; and if $m$
      is in $C_0$ then $\sigma(m)$ is in $C_0$.
%
\end{enumerate}
\end{theorem}
\begin{proof}
See~\cite{UresinDubois90}.
\end{proof}

This powerful theorem means that we can 
determine asynchronous correctness,
without having to reason about asynchronous processes.
We merely have to verify that the synchronous iteration converges,
and that certain conditions hold for the state space.
Unfortunately, these conditions may be rather tricky to apply
in practice---particularly if one wants to demonstrate that an
operator is \emph{not} an ACO. A further problem is that the low-level
way in which the conditions are stated makes it difficult
to understand the class of ACOs in general.

Several sufficient conditions are known
which imply that the criteria of Theorem~\ref{thm:aco-box} are
fulfilled.
Numeric examples include weighted maximum norms over Banach
spaces~\cite{FrommerSzyld00,ElTarazi82},
$P$-contractions~\cite{Baudet78},
paracontractions~\cite{ElsnerKoltrachtNeumann92,Pott98},
and isotone mappings~\cite{Miellou75}.
For discrete data, there are also results
about isotone mappings~\cite{UresinDubois90}.
All of these impose further requirements on the state space and
iteration;
the situation is especially painful for iterations on discrete
data, since none of the usual real-number apparatus 
is available: we do not have continuity, norms, or even subtraction.
The result of Section~\ref{sec:mainthm} is a necessary and
sufficient condition for iterations on discrete data, so
it covers iterations that are not necessarily isotone but even so
manage to converge to a unique fixed point.


\subsection{Ultrametric spaces}
\label{ssec:ultra}

An ultrametric is a particular kind of `distance' measurement that
differs in several important respects from more familiar examples,
but which also has useful applications.

A conventional metric space allows one to measure the distance between
two points as a real number, with certain intuitive properties being
fulfilled: a zero distance means the points are identical, distances
are symmetric, and the triangle inequality is valid.
In an ultrametric space, the triangle inequality is strengthened.
If $l$, $m$, and $n$ are three points in the space, then they must
form an isosceles triangle: two of the distances $d(l, m)$, $d(m, n)$
and $d(n, l)$ are the same.
Furthermore, the remaining distance can be no longer than the others,
so the triangle is `long and narrow' as opposed to `short and wide'.
For this reason, the spaces are sometimes called \newterm{isosceles spaces}.

\begin{definition}
An \newterm{ultrametric space} $(M, d, \Gamma)$ consists of a set $M$,
a totally
ordered set $\Gamma$ with least element $0$, and a function
$d: M \times M \rightarrow \Gamma$ such that
\begin{enumerate}
 \item $d(m, n) = 0$ if and only if $m = n$
 \item $d(m, n) = d(n, m)$
 \item $d(l, n) \le \max(d(l, m), d(m, n))$
\end{enumerate}
for all $l$, $m$ and $n$ in $M$.
\end{definition}

A canonical example is when
$M$ is the set of all strings over some alphabet.
For distinct $x$ and $y$ in $M$, let
\[
    m(x, y) = \min \setst{i \in \NN}{x_i \neq y_i}
\]
so $m$ yields the first index at which the two strings differ.
Then
\[
   d(x, y) = 
    \begin{cases}
      0            & x = y \\
      2^{-m(x, y)}  & x \neq y
    \end{cases}
\]
is an ultrametric distance function.

Indeed, this example is very close to being universal:
any ultrametric space $M$ is isometric to a space where the
elements are functions
from $\QQpos$ to $M$, and the distance is given by
\[
   d(f, g) = \sup \setst{q \in \QQpos}{f(q) \neq g(q)}.
\]
See \cite{Lemin02} for more details.
Some other examples come from Boolean algebra, where the
distance between two elements can be defined in terms of
their symmetric difference~\cite{PriessCrampeRibenboim96}.

If $h$ is a function from $M$ to $\Gamma\setminus\set{0}$, then
a distance function can be defined by
\[
   d_h(m, n) = 
      \begin{cases}
       0 & m = n\\
       \max(h(m), h(n)) & m \neq n.
      \end{cases}
\]
This is clearly an ultrametric.

\begin{definition}
In an ultrametric space $(M, d, \Gamma)$, the \newterm{ball} about a
point $m$ in $M$, of radius $r$ in $\Gamma$, is the set
\[
   \ball{r}{m} = \setst{n \in M}{d(m, n) \le r}.
\]
\end{definition}

The balls of an ultrametric space have some surprising properties.
Any point in a ball will serve as its center (see the lemma below).
If we have two balls $\ball{r}{n}$ and $\ball{s}{m}$, then
either one is a subset of the other, or they are disjoint.
Consequently, the set of all balls, ordered by inclusion, 
has a tree structure~\cite{Lemin03}.

\begin{lemma}
If $n$ is in $\ball{r}{m}$ then $\ball{r}{n} = \ball{r}{m}$.
\end{lemma}
\begin{proof}
If $n$ is in $\ball{r}{m}$ then $d(m, n) \le r$.
Then for any $z$ in $M$, 
\[
    d(m, z) \le \max(d(m, n), d(n, z))
\]
so if $d(n, z) \le r$ then $d(m, z) \le r$ as well; and the same
argument applies if $m$ and $n$ are interchanged. Hence the two balls
have the same content.
\end{proof}

We will need some notion of completeness of an ultrametric space, in
order to guarantee the existence of fixed points. Otherwise, it could
be that for certain iterations, the sequence of values converges to a
point that is not in the space.
The following definition
is sufficient
to ensure that the Banach fixed-point theorem actually yields
a fixed point.

\begin{definition}
An ultrametric space is \newterm{spherically complete} if every chain
of balls has nonempty intersection.
\end{definition}

Simple examples of spherically complete spaces include any space that is
finite, and any for which the image of $d$ is a finite subset of
$\Gamma$.
In both of these cases, any chain of balls is guaranteed to be finite,
and its intersection is then equal to the smallest ball in the chain.

The Banach theorem for ultrametric spaces can be made
to work for several different kinds of contracting operator.
The general idea is that, with respect to the ultrametric
distance,
each application of the operator brings points closer together.
Eventually, the entire space is contracted into a single point, which
is the desired unique fixed point.

\begin{definition}
Let $\sigma$ be a function from $M$ to $M$.
If $(M, d, \Gamma)$ is an ultrametric space, then $\sigma$ is (with
respect to $d$):
\begin{itemize}
\item a \newterm{contraction} if $d(\sigma(m), \sigma(n)) \le d(m, n)$
      for all $m$ and $n$ in $M$
\item a \newterm{strict contraction} if $d(\sigma(m), \sigma(n)) < d(m,
      n)$ for all distinct $m$ and $n$ in $M$
\item a \newterm{strict contraction on orbits} if
 $d(\sigma(m), \sigma^2(m)) < d(m, \sigma(m))$, or $m = \sigma(m)$,
 for all $m$ in $M$.
\end{itemize}
\end{definition}

\noindent
Note that if $\sigma$ is a strict contraction,
then it is necessarily a contraction that is strict on orbits.

\begin{theorem}
If $\sigma$ is a function from $M$ to $M$, and $(M, d, \Gamma)$ is a
spherically complete
ultrametric space with respect to which $\sigma$ is
a contraction that is strictly contracting on orbits,
then $\sigma$ has a unique fixed point.
\end{theorem}
\begin{proof}
See~\cite{PriessCrampe90} and~\cite{PriessCrampeRibenboim92}.
\end{proof}

We now demonstrate
that ultrametric spaces can be combined
via a product operation.
Furthermore, in the resulting space,
every ball is a box.
This provides the desired connection
with the theory of asynchronous iterations:
the balls about the fixed point will be precisely the
boxes demanded by the asynchronous iteration theorem.

\begin{definition}
Given ultrametric spaces $(M_i, d_i, \Gamma)$ for $1 \le i \le k$,
define the \newterm{ultrametric product space} $(M, d, \Gamma)$ by
\begin{align*}
 M & = \prod_{1 \le i \le k} M_i \\
 d(m, n) & = \max_{1 \le i \le k} d_i(m_i, n_i)
\end{align*}
where $m$ and $n$ are vectors in $M$.
\end{definition}

We will refer to $k$ as the \newterm{dimension} of $M$.
This usage is appropriate for the case when $M$ is a real or
complex vector space. When $M$ is discrete, its topological dimension
is zero; but here, we will carry on using the term `dimension'
for the number of components of the product.

\begin{lemma}
In an ultrametric product space, every ball is a box. That is,
if $(M, d, \Gamma)$ is the product of $(M_i, d_i, \Gamma)$
for $1 \le i \le k$,
then for any $m$ in $M$ and $r$ in $\Gamma$,
\[
   \ball{r}{m} = \prod_{1 \le i \le k} B_i
\]
where each set $B_i$ is a subset of $M_i$.
\end{lemma}
\begin{proof}
For each $i$, let $B_i$ be $\ball{r}{m_i}$ in $M_i$.
For any element $x$ of $M$, we have:
\begin{align*}
       & x \in \ball{r}{m} \\
  \iff & r \ge d(x, m) \\
  \iff & r \ge \max_{1 \le i \le k} d_i(x_i, m_i) \\
  \iff & \forall i: 1 \le i \le k \implies r \ge d_i(x_i, m_i) \\
  \iff & \forall i: 1 \le i \le k \implies x_i \in \ball{r}{m_i} \\
  \iff & \forall i: 1 \le i \le k \implies x_i \in B_i \\
  \iff & x \in \prod_{1 \le i \le k} B_i. \qedhere
\end{align*}

\end{proof}


\section{The main result}
\label{sec:mainthm}

This section is dedicated to proving the theorem below,
which provides a necessary and sufficient condition
for an operator to be an ACO,
in terms of an ultrametric structure on the state space.

\begin{theorem}
Let $M$ be a set, and $\sigma : M \rightarrow M$ a function from $M$
to $M$. Then $\sigma$ is an asynchronously contracting operator on $M$
if and only if there exists an ultrametric $d$ on $M$, with finite image,
and with respect to
which $\sigma$ is a contraction that is strict on orbits.
\end{theorem}

We will prove the two directions of this theorem separately.
In order to establish this result, we will need to use the
property that an operator on $M$ is asynchronously contracting if and only if
a series of nested boxes in $M$ exist, with certain properties.
The existence of an ultrametric will provide these boxes, and
conversely, given a series of boxes, we can define a suitable
ultrametric.

\begin{lemma}
If $(M, d, \Gamma)$ is a ultrametric space, with respect to which
$\sigma$ is a contraction that is strict on orbits, 
and $\Gamma$ is finite, then a series of
boxes exists that has the required properties.
\end{lemma}
\begin{proof}
From the theory of ultrametric spaces, the contraction conditions on
$\sigma$ provide that it has a unique fixed point in $M$; call this
$m^{\ast}$.
For every possible radius $r$ in $\Gamma$, there is a ball of
radius $r$ about $m^{\ast}$:
\[
   \ball{r}{m^{\ast}} = \setst{m \in M}{d(m^{\ast}, m) \le r}.
\]
These balls will be the required boxes.

Firstly, $\ball{0}{m^{\ast}} = \set{m^{\ast}}$, since no other points are at
distance zero from $m^{\ast}$ itself.

Next, due to finiteness of $\Gamma$,
there must be some minimal radius $k$ such that $\ball{k}{m^{\ast}} = M$.
Let $R$ be the set $\setst{r}{0 \le r \le k}$.
Clearly, if $r_1 < r_2$ for some $r_1$ and $r_2$ in $R$, then
$\ball{r_1}{m^{\ast}} \subseteq \ball{r_2}{m^{\ast}}$. Because some of these 
balls may coincide, despite having different radii, we define a subset
$S$ of $R$ such that
\begin{enumerate}
\item if $r$ is in $R$ then there is some $s$ in $S$ such that 
      $s \le r$ and $\ball{r}{m^{\ast}} = \ball{s}{m^{\ast}}$; and
\item for any $s_1$ and $s_2$ in $S$, $\ball{s_1}{m^{\ast}} \neq \ball{s_2}{m^{\ast}}$.
\end{enumerate}
This $S$ now yields the required sequence of boxes; it is well-defined
since $R$ is finite.

It remains to show that $\sigma$ fulfils the required property.
For any $m$ in $M$, we have the relationship
\[
    d(m, \sigma(m)) = d(m, m^{\ast})
\]
and, if $m$ is not equal to $\sigma(m)$,
\[
    d(m, \sigma(m)) > d(\sigma(m), \sigma^2(m)).
\]
Hence
\[
    d(m, m^{\ast}) > d(\sigma(m), m^{\ast})
\]
unless $m$ and $\sigma(m)$ are equal (in which case they are both
equal to $m^{\ast}$ itself).
Consequently, application of $\sigma$ always takes a point to a
more-inward ball, unless that point is the fixed point already.
\end{proof}

The next step is to prove the converse:
that if an operator is asynchronously contracting,
then there is an ultrametric with respect to which the operator
is a contraction that is strict on orbits.

\begin{lemma}
Let $M$ be a set endowed with an asynchronously contracting operator
$\sigma$. Then there is an ultrametric $d$ on $M$ with respect to
which $\sigma$ is a contraction that is strictly contracting on orbits.
\end{lemma}
\begin{proof}
Since $\sigma$ is an asynchronously contracting operator, there exists
a series of nested boxes with certain properties. We will define
an ultrametric distance function that uses this box structure.

Let $C_0$, $C_1$, \dots, $C_k$ be the box sequence, where $C_0$
is a singleton set, $C_k = M$, and $C_i \subset C_j$ whenever $0 \le i
< j \le k$.

For any point $m$ in $M$, we can find the index of the
innermost hypercube that contains $m$:
\[
   C(m) = \min\setst{i}{m \in C_i}.
\]
Thus $C(m^{\ast}) = 0$ if and only if $m^{\ast}$ is the unique fixed
point.
The required distance function is
\[
   d_C(m, n) =
    \begin{cases}
      0 & m = n \\
      \max( C(m), C(n) ) & m \neq n
    \end{cases}
\]
which is an ultrametric. Note that there are only finitely many
possible radii.

We can now show that the balls about $m^{\ast}$ with respect to $d$ are
precisely the given boxes $C_i$.
For any radius $r$,
\[
      m \in \ball{r}{m^{\ast}} 
 \iff  d(m, m^{\ast}) \le r 
 \iff  C(m) \le r 
 \iff   \min\setst{i}{m \in C_i} \le r 
 \iff  m \in C_r.
\]
Finally, we prove that $\sigma$ must be a contraction that is strictly
contracting on orbits. Note that $C(\sigma(m)) < C(m)$ for all $m$
other than $m^{\ast}$. Hence $d(\sigma(m), m^{\ast}) < d(m, m^{\ast})$ for such
$m$. Since $d(m, m^{\ast}) = d(m, \sigma(m))$ for all $m$, we have
\[
   d(m, \sigma(m)) > d(\sigma(m), \sigma^2(m))
\]
whenever $m \neq \sigma(m)$, so $\sigma$ is strictly contracting on
orbits.
Similarly, if $m$ and $n$ are two points in $M$, then $C(\sigma(m)) \le
C(m)$ and $C(\sigma(n)) \le C(n)$.
Therefore the larger of $C(\sigma(m))$ and $C(\sigma(n))$ cannot
exceed the larger of $C(m)$ and $C(n)$; we obtain
\[
   d(\sigma(m), \sigma(n)) \le d(m, n)
\]
as desired.
\end{proof}

This completes the proof.
We have established that an operator is asynchronously contracting
if and only if
it fulfils the ultrametric Banach fixed point theorem.
Furthermore, the degree of asynchrony
is given by the dimension of the ultrametric space.
This provides a convenient proof technique for
asynchronous iterations,
subsuming several other previously-known special conditions.

\section{Recovery of previous theorems}
\label{sec:prev-thms}

The general theorem of Section~\ref{sec:mainthm}
has many specific consequences
for particular classes of iteration.
Several of these have previously been
studied in the literature.
In this section, we relate previous
results to the new theory, thereby
demonstrating its generality.

\subsection{Paracontractions}

In the case of iterations
over real vectors,
there is a well-known theory of
paracontracting operators.
In the following, let $M$ be $\RRpos^n$, where $n \ge 1$.
For $v$ in $M$, let $\norm{v}$ be the vector whose $i$th entry
is the absolute value of $v_i$. For two vectors $v$ and $w$ in $M$,
say that $v \le w$ if for all $i$, $v_i \le w_i$.

\begin{definition}
A function $\sigma$ from $M$ to $M$ is called
a \newterm{paracontraction} if there exists
an $n$ by $n$ matrix $P$ with entries in $\RRpos$,
having spectral radius less than $1$,
and for which
\[ 
   \norm{\sigma(x) - \sigma(y)} \le P \norm{x - y}.
\]
\end{definition}

\begin{theorem}
If $\sigma$ is a paracontraction on $M$, then $(M, d)$ is
an ultrametric space with respect to which $\sigma$ is a strict
contraction, where
\[
  d(x, y) = \max_{1 \le i \le n} \alpha_i \norm{x - y}_i.
\]
\end{theorem}
\begin{proof}
If $\sigma$ is paracontracting, then it is a strict contraction
with respect to a weighted maximum norm~\cite{Baudet78,Bertsekas83}.
So the previous result suffices to prove this one.
\end{proof}

\subsection{Weighted maximum norms}

Suppose that the set $M$ is a product of sets
$M_1$ through $M_n$, each of which is equipped with
a real-valued norm $\norm{\cdot}_i$.
For any real numbers $\alpha_i$, a norm can be defined on $M$
by
\[
   \norm{x} = \max_i \alpha_i \norm{x_i}_i.
\]
An operator $\sigma$ on $M$ is \newterm{Lipschitz} if there
exists some $p$, with $0 \le p < 1$, such that
\[
  \norm{\sigma(x) - \sigma(y)} \le p \norm{x - y}
\]
for all $x$ and $y$ in $M$.

\begin{theorem}
If $\sigma$ is Lipschitz with respect to a weighted max-norm,
then it is also a strict contraction
with respect to an ultrametric distance on $M$, of dimension $n$.
\end{theorem}
\begin{proof}
For each $i$, let $d_i$ be the ultrametric on $M_i$ given by
\[
  d_i(x, y) = \alpha_i \norm{x - y}_i.
\]
\end{proof}

\subsection{Monotonic contractions}

Let $\le$ denote the direct product order on vectors in $(\RRpos)^n$,
so that
\[
   x \le y \iff \forall i: x_i \le_i y_i.
\]
A function $\sigma$ on $M$ is monotonic if
\[
   x \le y \implies \sigma(x) \le \sigma(y)
\]
for all $x$ and $y$ in $M$.
If $\sigma$ is a monotonic function,
with a unique fixed point $x^{\ast}$,
and
(other conditions) 
then it is an asynchronously contracting operator.

This conclusion also follows from the theorem of this paper.
To construct the required ultrametric, we can define ultrametrics
on each component, and then take their max-product.
The individual ultrametrics can be defined
in terms of the natural ordering on $\RRpos$.

Let $h(x) = 2^{-x}$. Note that for $x$ greater than or equal to zero,
$h(x)$ is in the range $(0,~1]$. Therefore, the function
\[
   d(x, y) = \begin{cases}
              0 & x = y \\
              \max(h(x), h(y)) & x \neq y \\
             \end{cases}
\]
is an ultrametric distance function on $\RRpos$;
and so the product of $n$ of these is also an ultrametric,
on $(\RRpos)^n$.
The conditions on $\sigma$ imply that
it is contractive
with respect to this ultrametric.

\section{Application to interdomain multipath routing}
\label{sec:multistable}

We will now see an extended example
of the use of the ultrametric theorem
to prove asynchronous safety of an iteration.
The iteration in question is simple to describe,
but has some unusual properties which make it
unsuitable for handling by previously-known asynchrony
theorems.

The problem comes from the selection of paths at
the interdomain level in network routing.
The various networks which combine to make the Internet
carry out path selection in a way which provides a great deal of local
autonomy: the paths for a given source-destination pair could 
in principle
be ranked arbitrarily.
Because of the local nature of preferences and decisions,
the overall routing outcome is not a global optimum,
as for shortest-path algorithms,
but is a Nash equilibrium between the
networks involved~\cite{GriffinShepherdWilfong02}.
This problem is also inherently distributed and asynchronous:
the private nature of policy means that the computation cannot
be performed centrally,
and synchronization on a global scale is infeasible.

In execution of the path selection process,
even synchronously, various anomalies appear which
would be impossible for shortest path algorithms.
It is \emph{not} necessarily the case that the paths
selected by a given node improve over time: a node is
perfectly capable of switching to a worse path if its
previous path becomes unavailable. A path could be lost
and then regained, possibly several times. Across the entire
network, it may be that at a given time step \emph{all} nodes
either stay the same or are forced to choose a worse path.
Nevertheless, an eventual fixed point can still (sometimes) be found,
even after all of these strange events have occurred.

Theoretical models of this situation include
the \newterm{stable paths problem}~\cite{GriffinShepherdWilfong02}. This is a
combinatorial game where the Nash equilibria
correspond to solutions of the interdomain routing problem.
An instance of the stable paths problem is given by:
\begin{itemize}
 \item A graph $G = (V, E)$ with a designated destination node $d$ in
       $V$.
 \item For each $v$ in $V$, a partial ranking of 
the simple paths in $G$ from $v$ to $d$. (That is, a total order on a
subset of these paths.) These are called the permitted paths.
\end{itemize}
The empty path (from $d$ to $d$) is always permitted.
A solution of the problem is a stable path assignment.
A \newterm{path assignment} associates each node in $V$ with at most
one of its permitted paths. 
A path assignment $\pi$ is \newterm{stable} if,
for each $v$ other than $d$ itself, $\pi(v)$ is identical
to the most-preferred path in $\setst{(v\,w)\pi(w)}{(v\,w) \in E}$,
and $\pi(d)$ is the empty path.

A stable path assignment is therefore a fixed point
of the myopic best-response iteration.
Stable path instances are known to exist with zero, one, or several
stable solutions. Certain sufficient criteria for the existence
of a unique stable solution are known, but may be NP-hard to check.
The iteration has previously been shown to be asynchronously
safe, with each node being responsible for managing its local state~\cite{GriffinShepherdWilfong02}.
The proof relies on a message-passing model with explicit queues,
rather than using the `asynchronous iteration' framework.
An alternative proof based on asynchronous iterations
did not succeed in demonstrating the desired result
\cite{Chau06}.

In this section, we extend the stable paths problem to the selection
of multiple paths.
(There is a version of the stable paths problem
in which mixed Nash equilibria are allowed: this is, however,
problem from the multiple-path problem presented here~\cite{HaxellWilfong08}.)
The proof uses the theory developed in this paper.
Consequently, this establishes a new correctness result for
multiple stable paths; and the result easily
extends to the asynchronous case.
This example had previously been treated using ultrametric spaces,
but without any proof of asynchronous correctness~\cite{Gurney09}.

A known sufficient condition for existence of a unique solution,
in the case of single-path selection,
is that the path preferences be determined according to
a strictly inflationary order~\cite{Sobrinho05,GriffinGurney08}.
That is, there is a total order on simple paths,
such that a path is strictly preferred to any extension of that same
path.

For multiple paths, we can make a similar definition.
Suppose that the set $\mathcal{P}$ of simple paths to $d$ has a preorder $\preceq$,
with $p \prec q$ indicating that $p$ is preferred to $q$.
(Recall that a preorder is a reflexive and transitive relation.)
If, for any path $p$ from $j$ to $k$, and any arc $(i\,j)$, we have
$p \prec (i\,j)p$, then the preferences are \newterm{strictly inflationary}.
For any subset $A$ of $\mathcal{P}$,
we can define
\[
    \min(A) = \setst{a \in A}{\forall b \in A: \neg(b \prec a)},
\]
the set of minimal elements in $A$.

We can now define the iteration.
The state space $M$ is a product of sets $M_i$ for $i$ in $V$.
Each $M_i$ is the powerset of the set of simple paths from $i$ to $d$.
Therefore, $M$ is isomorphic to the powerset of $\mathcal{P}$.
The iteration proceeds as follows:
\[
   \sigma(x)_i = 
     \begin{cases}
       \min\setst{(i\,j)p}{(i\,j) \in E, p \in x_j} 
         & i \neq d\\
       \set{\epsilon} & i = d
     \end{cases}
\]
where $\epsilon$ is the empty path.
So at each step,
a node
sees the paths which have been selected by its neighbors,
and from their extensions, chooses its `best' paths.
Note that
a node may \emph{lose} a path if its prefix ceases to be chosen
by a neighbor.

The distance function for the space
is defined by
\[
   d(m, n) = \max\setst{h(p)}{p \in (m \setminus n) \cup (n \setminus m)}
\]
for elements $m$ and $n$ of $M$,
where $h(p) = \setsize{\setst{q \in \mathcal{P}}{q \ge p}}$.
Note that if $p \prec q$, then $h(p) > h(q)$
If the set $m \setminus n \cup n \setminus m$ is empty, then
consistently
with the definition of $\max$, we take $d(m, n) = 0$.
This gives an ultrametric space.

To prove that the best-path selection process always terminates
in this case, and that it is asynchronously safe,
we need only show that $\sigma$ is a strict contraction.

\begin{theorem} \label{thm:strict-contr-mp}
With $M$, $d$ and $\sigma$ as defined above,
$\sigma$ is a strict contraction.
\end{theorem}
\begin{proof}
We must show that
\[
  d(\sigma(m), \sigma(n)) < d(m, n)
\]
for all distinct $m$ and $n$ in $M$.
If $\sigma(m) = \sigma(n)$ then there is nothing to prove,
so assume that they are different.
Then
\[
  d(\sigma(m), \sigma(n)) = h(p)
\]
for some $p$; without loss of generality,
$p$ is in $\sigma(m)$ but not $\sigma(n)$.

Note that $p$ cannot be the empty path, since
this will be present in both $\sigma(m)_d$ and $\sigma(n)_d$.
Therefore $p = (i\,j)q$ for some $q$ in $m_j$.
By the strict inflationary property,
we have $q \prec p$, and hence $h(q) > h(p)$.

Now, we will show that $q$ is not in $n_j$.
If it was, then $p$ would have been a candidate path for
$\sigma(n)_i$; but it was not selected.
The only reason for that to happen would be that
$\sigma(n)_i$ chose some better path $p'$ instead.
That $p'$ could not be in $\sigma(m)_i$, since
$\sigma(m)_i$ does contain a worse path, namely $p$;
so $p'$ is in the symmetric difference of
$\sigma(m)$ and $\sigma(n)$.
But if $p' \prec p$, then $h(p') > h(p)$, contradicting
the choice of $p$ as having greatest $h$ value among
all paths in the symmetric difference.

Therefore, $q$ is not in $n_j$.
Since it is in the symmetric difference of $m$ and $n$,
we have
\[
  d(\sigma(m), \sigma(n)) = h(p) < h(q) \le d(m, n)
\]
which establishes $\sigma$ as a strict contraction.
\end{proof}

This result demonstrates,
quite succinctly,
that the iterative path-finding process
always finds its unique fixed point,
and that this iteration can be implemented
asynchronously.

It is immediate from the definition that
the ultrametric space $(M, d)$ can be written
as a product in several ways.
Each of these corresponds to a mode of implementing $\sigma$ by
separate, asynchronously communicating processors.
Indeed, the maximal possible decomposition
for which $\sigma$ remains an ACO is for the
presence or absence of \emph{each individual path} to be
calculated by a separate processor:
\[
   (M, d) = \prod_{p \in \mathcal{P}} (M_p, d_p)
\]
where $M_p = \set{\epsilon, p}$ and $d_p(\epsilon, p) = h(p)$.

This is significantly more relaxed than current execution
models for interdomain routing, which admit asynchrony to the
extent that a node may compute paths to all of its destinations, separately
from any other node computing paths to all of its destinations.
The new proof demonstrates that we could
relax this model in several ways. The unit of computation could be:
\begin{itemize}
\item Paths for a given source (as now).
\item Paths for a given source and destination.
\item Paths for a given source, destination, and next-hop (at the
  router level or at the autonomous system level).
\item Paths originating from a given geographic region, or
  a given class of autonomous systems.
\item Paths of a given level of preference.
\end{itemize}
In many of these cases, we admit the possibility that intermediate
states of the routing system could contain `inconsistencies'. For
example, it could be that paths $p$ and $p'$, for the same source and
destination, with $p$ preferred to $p'$, might be simultaneously
present in the system state. While this may seem unusual, we have
proved that it will not harm eventual convergence to a consistent
state. Already, in existing routing, we consider it normal for a node
to be using a path that is transiently inconsistent with its next-hop
neighbor's choice: this is essentially the same idea.
(Note that while the proof of Theorem~\ref{thm:strict-contr-mp} relies on
such states not occuring, this result applies to the synchronous
execution only: asynchronous executions are allowed to behave more wildly.)

We have therefore revealed that routing protocol designs can be
extended in two ways, compared to the present state of the art.
Firstly, as long as the above rules are followed, we can move from
single-path to multipath routing with confidence.
Secondly, the computation of routes by asynchronous processes can be
accomplished in many other ways from the method previously known to be
safe.

\section{Application to logic programming}
\label{sec:logic}

Consider logic programs consisting of clauses
\[
  A \leftarrow L_1, \dots, L_n
\]
where each $A$ is an atom,
and each $L_i$ is either an atom or the negation of an atom.
(There may be no atoms in the list, in which case the clause
is called a `fact'.)

Denote the Herbrand base of a program $P$ by $B_P$.

A program $P$ is said to be \newterm{locally stratified} if
there is a mapping $\rho : B_P \longrightarrow \Gamma$ into
a totally ordered set $\Gamma$, such that for each clause
$A \leftarrow L_1, \dots L_n$ in $P$, we have
\begin{enumerate}
\item $\rho(A) \ge \rho(B)$ for all positive atoms $B$, and
\item $\rho(A) > \rho(B)$ for all negative atoms $B$.
\end{enumerate}

An \newterm{interpretation} of $P$ is a map $I$ from $B_P$
to $\set{\mathrm{true}, \mathrm{false}}$.

The \newterm{immediate consequence operator} $T_P$
is defined as (TODO).

Define an ultrametric $d$ on interpretations by
\[
d(I, J) = \min\setst{\rho(A)}{A \in (I \setminus J)
  \cup (J \setminus I)}.
\]
This is spherically complete, and $T_P$ is a strict
contraction on it.
This demonstrates the existence of a unique fixed point,
corresponding to the perfect model semantics of the program.

By recourse to the above theorem,
we have also shown that execution of the logic program
can be spread among independent processors.
Indeed, any term in the Herbrand base could be the
responsibility of a separate processor.

\section{Conclusion}
\label{sec:conclusion}

This paper
has established a connection
between the theories of asynchronous iterations
and ultrametric spaces.
It provides a complete characterization
of asynchronously contracting operators,
and hence gives new proof techniques
by which asynchronous safety can be proved 
(or disproved)
in particular cases.

It is hoped
that
this formalism will be useful in understanding
the class of ACOs,
and in finding new examples,
particularly those which are not numeric.
There is also work to be done in treating
more complex families of iteration, such as those
possessing multiple fixed points.
The study of ultrametric embeddings and representations
may also provide insight into
asynchronous contractions and their properties.

In the realm of networking practice,
the suggestion of alternative routing architectures is tantalising.
Further research will be needed in order to establish which such
designs are desirable.

\bibliography{async}

\end{document}